\documentclass[lettersize, journal]{IEEEtran}  
\IEEEoverridecommandlockouts
\usepackage{lineno}
\usepackage{hyperref}
\usepackage{cite}
\usepackage{amsmath,amssymb,amsfonts}
\usepackage{amsmath}
\usepackage{amsthm}

\usepackage{graphicx}
\usepackage{textcomp}
\usepackage{xcolor}
\usepackage{graphicx}
\usepackage{float}
\usepackage{subfigure}
\usepackage{amsmath,mleftright}
\usepackage{amsfonts,amssymb}
\usepackage{mathrsfs}
\usepackage{mathtools}
\usepackage{algorithm}
\usepackage{algorithmicx}
\usepackage{algpseudocode}
\usepackage{bm}
\usepackage{multirow}
\usepackage{array}
\usepackage{amssymb}
\usepackage{amsmath}
\usepackage{cite}
\usepackage{url}
\usepackage{xcolor}
\usepackage{cite,graphicx,amsmath,amssymb}
\usepackage{subfigure}
\usepackage{fancyhdr}
\usepackage{mdwmath}
\usepackage{mdwtab}
\usepackage{caption}
\usepackage{amsthm}
\usepackage{setspace}
\usepackage{bm}
\usepackage{algorithm}
\usepackage{algpseudocode}
\usepackage{mathtools}
\usepackage{dsfont}
\usepackage{bbm}
\newtheorem{remark}{Remark}
\newtheorem{theorem}{Theorem}

\newtheorem{lemma}{Lemma}

\newtheorem{corollary}{Corollary}

\makeatletter
\newcommand{\biggg}{\bBigg@{3}}
\newcommand{\Biggg}{\bBigg@{3.5}}
\makeatother
\def\BibTeX{{\rm B\kern-.05em{\sc i\kern-.025em b}\kern-.08em
    T\kern-.1667em\lower.7ex\hbox{E}\kern-.125emX}}

\begin{document}
\title{Robust Beamforming for Pinching-Antenna Systems}
\author{
\author{
Mingjun Sun,~\IEEEmembership{Graduate Student Member,~IEEE,} Chongjun Ouyang,~\IEEEmembership{Member,~IEEE,}\\ Shaochuan Wu,~\IEEEmembership{Senior Member,~IEEE,} and Yuanwei Liu,~\IEEEmembership{Fellow,~IEEE}
\thanks{Mingjun Sun and Shaochuan Wu are with the School of Electronics and Information Engineering, Harbin Institute of Technology, Harbin 150001, China (e-mail: sunmj@stu.hit.edu.cn; scwu@hit.edu.cn).}
\thanks{Chongjun Ouyang is with the School of Electronic Engineering and Computer Science, Queen Mary University of London, London, E1 4NS, U.K. (e-mail: c.ouyang@qmul.ac.uk).}
\thanks{Yuanwei Liu is with the Department of Electrical and Electronic Engineering, The University of Hong Kong, Hong Kong (e-mail: yuanwei@hku.hk).}
}}
\maketitle
\begin{abstract}
Pinching-antenna system (PASS) mitigates large-scale path loss by enabling flexible placement of pinching antennas (PAs) along the dielectric waveguide. However, most existing studies assume perfect channel state information (CSI), overlooking the impact of channel uncertainty. This paper addresses this gap by proposing a robust beamforming framework for both lossy and lossless waveguides.
For baseband beamforming, the lossy case yields an second-order cone programming-based solution, while the lossless case admits a closed-form solution via maximum ratio transmission. The PAs' positions in both cases are optimized through the Gauss-Seidel-based method. Numerical results validate the effectiveness of the proposed algorithm and demonstrate that PASS exhibits superior robustness against channel uncertainty compared with conventional fixed-antenna systems. Notably, its worst-case achievable rate can even exceed the fixed-antenna baseline under perfect CSI.

\end{abstract} 
\section{Introduction}
Flexible-antenna architectures such as reconfigurable intelligent surfaces \cite{Wu2020}, movable antennas \cite{zhu2023movable}, and fluid antennas \cite{wong2020fluid} have challenged the notion of uncontrollable channels, introducing new degrees of freedom (DoFs) for multiple-input multiple-output (MIMO) enhancement. These designs boost rates by reconfiguring effective channel gains, mainly addressing small-scale fading. However, large-scale path loss remains a dominant source of attenuation. To address this, the pinching-antenna system (PASS) was proposed \cite{liu2025pinchingantenna}, which extends a waveguide and activates pinching antennas (PAs) at arbitrary positions along it to establish strong line-of-sight (LoS) links, thereby mitigating large-scale path loss \cite{suzuki2022pinching}.

Recently, extensive studies have been investigated on PASS.
The signal and system model was first introduced in \cite{ding2024flexible}, where its performance under various waveguide and PA configurations was analyzed. The optimal number and spacing of PAs were studied in \cite{ouyang2025array} by omitting the impact of in-waveguide propagation loss, while \cite{tyrovolas2025performanceanalysis} evaluated system outage probability when such loss is present. Additionally, \cite{LoS} explored the favorable impact of LoS blockage on suppressing co-channel interference.
Compared with fixed-antenna system, PASS provides additional spatial DoFs thorough pinching beamforming.
For single-waveguide scenarios, the PA placement algorithm has been proposed \cite{wang2024antenna,xu2024rate}. 
For multiple-waveguide scenarios, several joint optimization algorithms for baseband and pinching beamforming have been proposed for various objectives, including achievable rate maximization \cite{bereyhi2025mimopass,Mingjun}, power minimization \cite{wang2025modeling}, and max-min fairness\cite{KitPASS}. Beyond classical optimization methods, AI-based joint beamforming designs have also been investigated. Early efforts in this direction appeared in the context of movable antennas \cite{hanguo1,hanguo3}. More recently, the authors in \cite{Guojia} proposed a graph neural network (GNN)-based algorithm to efficiently handle the highly nonconvex and coupled optimization problems. In addition, \cite{hanguo2} incorporated the receiver antenna position into the optimization and introduced CaMPASS-Net to effectively tackle the point-to-point MIMO capacity maximization problem.

However, the above studies all assume known channel state information (CSI). Although existing works \cite{xiao2025channelestimation, zhou2025channel} achieve relatively high channel estimation accuracy for PASS, perfect CSI is not available. The impact of channel uncertainty has not been examined. To address this gap, this paper investigates robust beamforming under two representative channel uncertainty models: norm-bounded \cite{Robust1} and probabilistic error models \cite{lemma}.
The former assumes the channel error lies within a bounded spherical region, while the latter treats it as a complex Gaussian random variable. 
These lead to two corresponding formulations: worst-case robust optimization and chance-constrained optimization, aiming to maximize the worst-case and nonoutage achievable rates (ARs), respectively. These two metrics have been extensively studied in \cite{Robust1,worstcase,lemma,Outage1}. Another line of robustness aims to optimize the long–term (average) quality of service (QoS), which leads to a stochastic optimization problem\cite{stochastic1}. This third strand is relatively orthogonal to the above and is therefore not considered in this work. Note that for PASS, the in-waveguide channel is deterministic, and uncertainty arises solely from the wireless propagation outside the waveguide.

\begin{figure}[!t]
 \centering
\setlength{\abovecaptionskip}{0pt}
\includegraphics[height=0.21\textwidth]{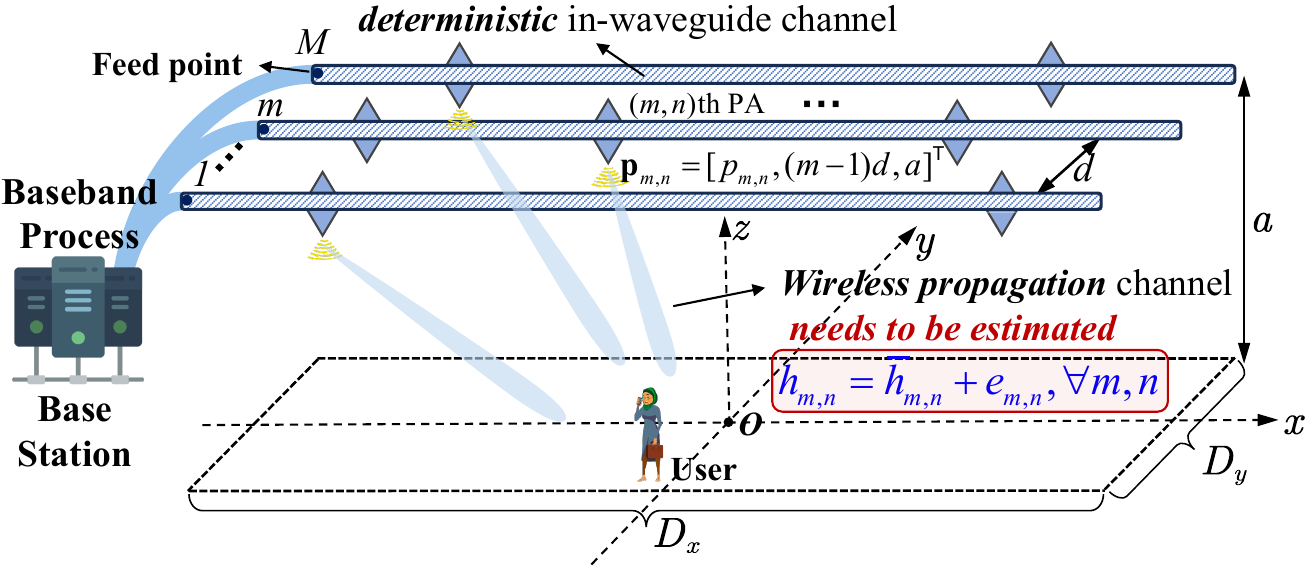}
\caption{Downlink PASS with imperfect CSI.}
\label{Figure: System_Model}
\vspace{-10pt}
\end{figure}

The main contributions of this article are summarized as follows: 
To address both the worst-case and the chance-constrained optimization problems, we first establish their equivalence, enabling a unified solution framework. We then propose a robust beamforming algorithm that alternately optimizes the baseband and pinching beamformers and guarantees convergence. For lossy waveguides, the baseband beamformer is obtained via second-order cone programming (SOCP), while for lossless case, it follows maximum ratio transmission (MRT). For the pinching beamformer, we employ
a Gauss-Seidel-based one-dimensional (GS1D) search to update the PAs' positions. 
Simulation results validate the proposed algorithm's effectiveness and demonstrate that PASS exhibits strong robustness to channel uncertainty, even surpassing fixed-antenna systems with perfect CSI in worst-case performance.

\section{System Model}\label{Section: System Model}
Consider a downlink PASS, where a base station (BS) serves a single-antenna user\footnote{This framework can, in principle, be extended to a multiuser setting under time-division multiple access (TDMA). As an initial exploration of robust PASS, however, the present work focuses on the single-user case. A comprehensive multiuser robust beamforming design that explicitly accounts for inter-user interference is left for future work.} using $M$ pinched waveguides, as illustrated in {\figurename} {\ref{Figure: System_Model}}. 
In each transmission slot, the BS simultaneously activates $N$ PAs per waveguide. 
Assuming the user is randomly located within a predefined two-dimensional rectangular area of size $D_x\times D_y$, with its position denoted as ${\mathbf{u}}=[x, y, 0]^{\mathsf{T}}$.
The $M$ waveguides\footnote{The proposed framework can also be readily simplified to apply to the single-waveguide scenario.} are deployed at height $a$ along the $x$-axis and uniformly spaced by $d$ along the $y$-axis. The length of the $m$th waveguide is denoted by $L_m$. The location of the $n$th PA on the $m$th waveguide (i.e., the $(m,n)$th PA) is given by ${\mathbf{p}}_{m,n}=[p_{m,n},-\frac{D_y}{2}+(m-1)d,a]^{\mathsf{T}}$, where $p_{m,n}$ indicates PA's position along the $x$-axis. 

\subsection{Signal Model}
For downlink transmission, the signal received by the user from all PAs can be expressed as follows:
\begin{align}\label{DL_Signal_Model_Initial}
y=\sum_{m=1}^{M}\sum_{n=1}^{N}h_{m,n}x_{m,n}+z,
\end{align}
where $x_{m,n}$ represents the signal radiated from the $(m,n)$th PA, $h_{m,n}$ is the corresponding channel coefficient, and $z\sim{\mathcal{CN}}(0,\sigma_k^2)$ denotes additive white Gaussian noise with noise power $\sigma^2$. Accounting for in-waveguide propagation from the BS to the PA, the signal $x_{m,n}$ is expressed as follows:
\begin{align}\label{Inner_Waveguide_Signal}
x_{m,n}=\eta_{m,n}^{\frac{1}{2}}{\rm{e}}^{-{\rm{j}}k_{\rm{g}}\left\lvert p_{m,n}-o_m\right\rvert}x_{m},
\end{align}
where $k_{\rm{g}}=\frac{2\pi}{\lambda_{\rm{g}}}$ is the waveguide wavenumber with $\lambda_{\rm{g}}$ representing the waveguide wavelength, $o_m$ is the signal feeding position on the $m$th waveguide, $x_m$ denotes the signal at the feeding point, and $\eta_{m,n}$ captures power loss along the dielectric waveguide. Following \cite{wang2025modeling}, the radiated power can be controlled by adjusting the coupling length between each PA and the waveguide. Based on the equal power model in \cite[Eq. (20)]{wang2025modeling}, we model $\eta_{m,n}$ as follows:
\begin{align}
\eta_{m,n}=10^{-\frac{\kappa\left\lvert p_{m,n}-o_m\right\rvert}{10}}/N.
\end{align}
where $\kappa\geq0$ denotes the average attenuation (in dB/m) along the waveguide \cite{yeh2008essence}. We note that $\kappa=0$ corresponds to a lossless dielectric waveguide. In practice, propagation loss accumulates with waveguide length, so longer waveguides cause greater performance degradation; see \cite{yanqing1,yanqing2} for more detailed analyses.

For compactness, we rewrite \eqref{DL_Signal_Model_Initial} as follows:
\begin{align}\label{DL_Signal_Model}
y={\mathbf{h}}^{\mathsf{H}}({\mathbf{P}}){\mathbf{G}}({\mathbf{P}}){\mathbf{x}}+z.
\end{align}
Here, ${\mathbf{P}}=[{\mathbf{p}}_{1},\ldots,{\mathbf{p}}_{M}]\in{\mathbbmss{R}}^{N\times M}$ collects the $x$-axis positions of all PAs, where ${\mathbf{p}}_{m}=[p_{m,1},\ldots,p_{m,N}]^{\mathsf{T}}\in{\mathbbmss{R}}^{N\times 1}$ stores the positions of the PAs along the $m$th waveguide. The channel vector is given by
${\mathbf{h}}({\mathbf{P}})=[\mathbf{h}_{\rm f}^{\mathsf{T}}({\mathbf{p}}_{1}),\ldots,\mathbf{h}_{\rm f}^{\mathsf{T}}({\mathbf{p}}_{M})]^{\mathsf{H}}\in{\mathbbmss{C}}^{MN\times1}$, where $\mathbf{h}_{\rm f}({\mathbf{p}}_{m})=[h_{m,1},\ldots,h_{m,N}]^{\mathsf{T}}\in{\mathbbmss{C}}^{N\times 1}$ denotes the channel vector from the user to the $m$th waveguide. The block diagonal matrix ${\mathbf{G}}({\mathbf{P}})={\rm{BlkDiag}}({\mathbf{g}}({\mathbf{p}}_{1});\ldots;{\mathbf{g}}({\mathbf{p}}_{M}))\in{\mathbbmss{C}}^{MN\times M}$ characterizes wave propagation within the waveguides, where ${\mathbf{g}}({\mathbf{p}}_{m})=\left[\eta_{m,1}^{\frac{1}{2}}{\rm{e}}^{-{\rm{j}}k_{\rm{g}}\left\lvert p_{m,1}-o_m\right\rvert},\ldots,\eta_{m,N}^{\frac{1}{2}}{\rm{e}}^{-{\rm{j}}k_{\rm{g}}\left\lvert p_{m,N}-o_m\right\rvert}\right]^{\mathsf{T}}\in{\mathbbmss{C}}^{N\times 1}$. Finally, ${\mathbf{x}}=[x_1,\ldots,x_M]^{\mathsf{T}}={\mathbf{w}}s\in{\mathbbmss{C}}^{M\times1}$ denotes the signal after baseband beamforming, where ${\mathbf{w}}\in{\mathbbmss{C}}^{M\times 1}$ is the baseband beamforming vector, and $s\sim{\mathcal{CN}}(0,1)$ is the transmit symbol.
Accordingly, the received signal-to-noise ratio (SNR) at the user is given by
$\gamma=\lvert{\mathbf{h}}^{\mathsf{H}}({\mathbf{P}}){\mathbf{G}}({\mathbf{P}}){\mathbf{w}}\rvert^2 /{\sigma^2}$.

\subsection{Channel Uncertainty Model}
In PASS, the overall channel comprises a deterministic in-waveguide component ${\mathbf{G}}({\mathbf{P}})$ and a wireless propagation component ${\mathbf{h}}({\mathbf{P}})$, where only the latter requires estimation\footnote{When a LoS path exists, user-localization-based spherical-wave LoS reconstruction can be used for channel estimation, whereas in purely nonline-of-sight (NLoS) or mixed LoS/NLoS conditions, uplink pilot-based training can be employed for more accurate channel acquisition. To avoid the prohibitive complexity of estimating channels for the theoretically infinite set of continuous PAs' positions, one may resort to interpolation-based reconstruction or predefine a reduced candidate set of PAs' positions tailored to the specific scenario.}. Due to the presence of receiver noise, limited pilot overhead, and model mismatch, estimation errors are inevitably introduced \cite{xiao2025channelestimation, zhou2025channel}. To characterize such imperfections, the actual wireless propagation channel vector is modeled as follows\footnote{The channel can be modeled as LoS, NLoS, or a mixture of both.}:
\begin{align}\label{Channel_Uncertainty}
\mathbf{h}(\mathbf{P})= \bar{\mathbf{h}}(\mathbf{P}) + \mathbf{e}(\mathbf{P}),
\end{align}
where $\bar{\mathbf{h}}({\mathbf{P}})=[\bar{\mathbf{h}}_{\rm f}^{\mathsf{T}}({\mathbf{p}}_{1}),\ldots,\bar{\mathbf{h}}_{\rm f}^{\mathsf{T}}({\mathbf{p}}_{M})]^{\mathsf{H}}\in{\mathbbmss{C}}^{MN\times1}$ is the estimated channel, with $\bar{\mathbf{h}}_{\rm f}({\mathbf{p}}_{m})=[\bar{h}_{m,1},\ldots,\bar{h}_{m,N}]^{\mathsf{T}}\in{\mathbbmss{C}}^{N\times 1}$, and $\mathbf{e}(\mathbf{P})\in{\mathbbmss{C}}^{MN\times1}$ denotes the corresponding estimation error. 
In general, the uncertainty of $\mathbf{e}(\mathbf{P})$ can be modeled in two ways: the \emph{norm-bounded error} model\cite{Robust1} and the \emph{probabilistic error} model\cite{lemma}.
In the former, the error vector lies within a bounded uncertainty set, typically defined as $\|\mathbf{e}(\mathbf{P})\|_{\rm 2} \leq \delta$, where $\delta >0$ denotes the error bound. In the latter, $\mathbf{e}(\mathbf{P})$ is treated as a random vector, and is commonly modeled as a circularly symmetric complex Gaussian variable with zero mean and a known covariance matrix $\epsilon^2 \mathbf{I}$, i.e., $\mathbf{e}(\mathbf{P}) \sim \mathcal{CN}(\mathbf{0}, \epsilon^2 \mathbf{I})$.

\subsection{Problem Formulation}
We consider two types of PA activation schemes: \emph{discrete} and \emph{continuous}.
For the former, all PAs on each waveguide $m$ are assumed to be activatable only at $Q$ pre-configured discrete positions, forming a discrete feasible set $\mathcal{S}_m = \left\{o_m+\frac{L_m}{Q-1}(q-1)| q=1,\ldots,Q\right\}$. For the latter, all PAs can be activated at any position along the waveguide.
This paper aims to jointly optimize the baseband and pinching beamformers to maximize the user's AR, equivalently maximizing the received SNR.
Based on the adopted channel uncertainty model, the problem can be formulated in two forms:
\subsubsection{Worst-Case Robust Optimization}
To account for all possible channel uncertainties, the problem leads to a semi-infinite optimization formulation, in which the objective is to maximize the worst-case received SNR over the entire uncertainty set:
\begin{subequations}\label{DL_SNR_Problem}
\begin{align}
(\mathcal{P}_1)\max_{{\mathbf{w}},\,{\mathbf{P}}}  & \min_{\|\mathbf{e}(\mathbf{P})\|_2 \leq \delta}~ \left| \mathbf{h}^{\mathsf{H}}(\mathbf{P}) \mathbf{G}(\mathbf{P}) \mathbf{w} \right| \\
\text{s.t.} & \|\mathbf{w}\|_2^2 \leq P_t, \label{dl_c3}\\
& |p_{m,n} - p_{m,n'}| \geq \Delta_{\min}, \forall m,~ n \ne n', \label{dl_c1}\\
& p_{m,n} \!\in \!\mathcal{S}_m \quad \text{or} \quad p_{m,n} \!- \!o_m \!\in\! [0, L_m], \forall m, n. \label{dl_c2}
\end{align}
\end{subequations}
where $P_t$ denotes transmit power budget, and $\Delta_{\rm min}>0$ ensures minimum spacing between adjacent PAs to avoid mutual coupling. 

\subsubsection{Chance-Constrained Optimization}
Considering that $\mathbf{e}(\mathbf{P})$ is a random variable, the received SNR also becomes random. Therefore, the problem can be formulated as a chance-constrained optimization:
\begin{subequations}\label{DL_SNR_Problem_chance}
\begin{align}
(\mathcal{P}_2)~\max_{{\mathbf{w}},{\mathbf{P}},\Gamma}&~ \Gamma\\
{\rm s.t.}~ & \Pr\left\{\lvert{\mathbf{h}}^{\mathsf{H}}({\mathbf{P}}){\mathbf{G}}({\mathbf{P}})\mathbf{w}\rvert \geq \Gamma\right\} \geq \rho \label{chance_C}\\ 
& \eqref{dl_c3},\eqref{dl_c1},\eqref{dl_c2},\nonumber
\end{align}
\end{subequations}
where $\Gamma = \sigma\sqrt{\Gamma_0}$ with $\Gamma_0$ denoting the nonoutage received SNR, $\Pr\{\cdot\}$ denotes the probability operator, and $\rho \in (0,1]$ specifies the required nonoutage probability.

\section{The Proposed Robust Beamforming Algorithm}
This section presents the solution procedure for problem $\mathcal{P}_1$. A lemma is then established to equivalently reformulate problem $\mathcal{P}_2$ into the same structure as $\mathcal{P}_1$, thereby enabling the use of the same solution approach.
Due to the semi-infiniteness of $\mathcal{P}_1$, the optimal robust pinching beamformer cannot be obtained directly. Therefore, it is necessary to reformulate the problem into a more tractable form.

By substituting \eqref{Channel_Uncertainty} into the objective function of \eqref{DL_SNR_Problem} and applying the triangle inequality, we have:
\begingroup              
\setlength{\abovedisplayskip}{2pt}
\setlength{\belowdisplayskip}{2pt}
\begin{align}
f(\mathbf{e}) &= \left\lvert\bar{\mathbf{h}}^{\mathsf{H}}({\mathbf{P}}){\mathbf{G}}({\mathbf{P}})\mathbf{w} + \mathbf{e}^{\mathsf{H}}({\mathbf{P}}){\mathbf{G}}({\mathbf{P}})\mathbf{w}\right\rvert \\\nonumber
&\geq \left\lvert\lvert \bar{\mathbf{h}}^{\mathsf{H}}({\mathbf{P}}){\mathbf{G}}({\mathbf{P}})\mathbf{w}\rvert - \lvert \mathbf{e}^{\mathsf{H}}({\mathbf{P}}){\mathbf{G}}({\mathbf{P}})\mathbf{w}\rvert\right\rvert,
\end{align}
\endgroup
where the equality holds when $\bar{\mathbf{h}}^{\mathsf{H}}({\mathbf{P}}){\mathbf{G}}({\mathbf{P}})\mathbf{w}$ and $\mathbf{e}^{\mathsf{H}}({\mathbf{P}}){\mathbf{G}}({\mathbf{P}})\mathbf{w}$ are in opposite phase. To further eliminate the uncertainty vector $\mathbf{e}$, the Cauchy-Schwarz inequality is applied, yielding:
\begingroup              
\setlength{\abovedisplayskip}{2pt}
\setlength{\belowdisplayskip}{2pt}
\begin{align}
\lvert \mathbf{e}^{\mathsf{H}}({\mathbf{P}}){\mathbf{G}}({\mathbf{P}})\mathbf{w}\rvert \!\leq \!\|\mathbf{e}({\mathbf{P}})\|_{\rm 2} \| {\mathbf{G}}({\mathbf{P}})\mathbf{w}\|_{\rm 2}  \!\leq \!\delta\| {\mathbf{G}}({\mathbf{P}})\mathbf{w}\|_{\rm 2}.
\end{align}
\endgroup
Moreover, it is easy to obtain that when $\lvert \bar{\mathbf{h}}^{\mathsf{H}}({\mathbf{P}}){\mathbf{G}}({\mathbf{P}})\mathbf{w}\rvert > \delta\| {\mathbf{G}}({\mathbf{P}})\mathbf{w}\|_{\rm 2}$ and $\mathbf{e}({\mathbf{P}}) = -\frac{\delta{\mathbf{G}}({\mathbf{P}})\mathbf{w}}{\| {\mathbf{G}}({\mathbf{P}})\mathbf{w}\|_{\rm 2}}\mathrm{e}^{j\angle\{\bar{\mathbf{h}}^{\mathsf{H}}({\mathbf{P}}){\mathbf{G}}({\mathbf{P}})\mathbf{w}\}}$, the worst-case $f(\mathbf{e})$ simplifies to $\left\lvert\bar{\mathbf{h}}^{\mathsf{H}}({\mathbf{P}}){\mathbf{G}}({\mathbf{P}})\mathbf{w} + \mathbf{e}^{\mathsf{H}}({\mathbf{P}}){\mathbf{G}}({\mathbf{P}})\mathbf{w}\right\rvert = \lvert \bar{\mathbf{h}}^{\mathsf{H}}({\mathbf{P}}){\mathbf{G}}({\mathbf{P}})\mathbf{w}\rvert - \delta\| {\mathbf{G}}({\mathbf{P}})\mathbf{w}\|_{\rm 2}$. In particular, when the error bound $\delta$ is large enough such that $\lvert \bar{\mathbf{h}}^{\mathsf{H}}({\mathbf{P}}){\mathbf{G}}({\mathbf{P}})\mathbf{w}\rvert \leq \delta\| {\mathbf{G}}({\mathbf{P}})\mathbf{w}\|_{\rm 2}$, the worst-case SNR drops to zero, rendering the robust beamformer ineffective.
Therefore, we have
\begin{align}
&\min_{\|\mathbf{e}(\mathbf{P})\|_{\rm 2} \leq \delta} \left\lvert\lvert \bar{\mathbf{h}}^{\mathsf{H}}({\mathbf{P}}){\mathbf{G}}({\mathbf{P}})\mathbf{w}\rvert - \lvert \mathbf{e}^{\mathsf{H}}({\mathbf{P}}){\mathbf{G}}({\mathbf{P}})\mathbf{w}\rvert\right\rvert \\ \nonumber
&\quad \quad= \max\left \{0, \lvert \bar{\mathbf{h}}^{\mathsf{H}}({\mathbf{P}}){\mathbf{G}}({\mathbf{P}})\mathbf{w}\rvert - \delta\| {\mathbf{G}}({\mathbf{P}})\mathbf{w}\|_{\rm 2}\right\}.
\end{align}
We focus on the non-zero case, which serves as a lower bound for $f(\mathbf{e})$. Accordingly, the original problem $\mathcal{P}_1$ can be equivalently reformulated as:
\begin{subequations}\label{DL_SNR_Problem_2}
\begin{align}
(\mathcal{P}_{\rm new})~\max_{{\mathbf{w}},{\mathbf{P}}}&~\lvert \bar{\mathbf{h}}^{\mathsf{H}}({\mathbf{P}}){\mathbf{G}}({\mathbf{P}})\mathbf{w}\rvert - \delta\| {\mathbf{G}}({\mathbf{P}})\mathbf{w}\|_{\rm 2}\\
{\rm{s.t.}}&~ \eqref{dl_c3},\eqref{dl_c1},\eqref{dl_c2},\nonumber
\end{align}
\end{subequations}
which is now a deterministic optimization problem. However, the optimization variables $\mathbf{w}$ and $\mathbf{P}$ are coupled, and due to the presence of the exponential form, the objective function exhibits multimodal characteristics, making the problem difficult to solve. 

To establish the equivalence between $\mathcal{P}_2$ and $\mathcal{P}_1$, we follow a similar approach in \cite{lemma} and present the following lemma.
\begin{lemma}\label{lemma1}
Given a nonoutage probability $\rho$, problems $\mathcal{P}_2$ and $\mathcal{P}_1$ are equivalent if the norm bound $\delta$ of the uncertainty set and the standard deviation $\epsilon$ of the random channel error at each PA's location satisfy the following condition:
\begin{align}\label{condition}
\delta = \epsilon \sqrt{-\ln(1-\rho)}.
\end{align}
\end{lemma}
\begin{proof}
A brief proof is provided in the Appendix \ref{Appendix_A}.
\end{proof}
\begin{remark}
Under condition \eqref{condition}, problems $\mathcal{P}_1$ and $\mathcal{P}_2$ yield the same joint beamforming solution. Even when the condition is not satisfied, $\mathcal{P}_2$ can still be tackled using the same approach as $\mathcal{P}_1$ through the parameter transformation in \eqref{condition}. Note that this equivalence relies on the conservative approximation bound we adopt for the chance constraint \eqref{chance_C}, which renders $\mathcal{P}_2$ tractable---particularly allowing it to be solved within the same framework as $\mathcal{P}_1$. A tighter approximation bound could further improve the nonoutage probability \cite{Outage1}.
\end{remark}

To address $\mathcal{P}_{\rm new}$, we employ an alternating optimization (AO) scheme to efficiently obtain a stable solution. In particular, the global optimum (worst-case AR) is necessarily no smaller than the value achieved by this solution.
\subsection{Baseband Beamformer}
Given the PAs' positions $\mathbf{P}$, we omit the notation $(\mathbf{P})$ for brevity. Problem $\mathcal{P}_{\rm new}$ can be rewritten as:
\begin{subequations}\label{DL_SNR_Problem_w}
\begin{align}
\min_{{\mathbf{w}}}&~\delta\| {\mathbf{G}}\mathbf{w}\|_{\rm 2}-\lvert \bar{\mathbf{h}}^{\mathsf{H}}{\mathbf{G}}\mathbf{w}\rvert\\
{\rm{s.t.}}&~ \eqref{dl_c3},\nonumber
\end{align}
\end{subequations}
Despite its nonconvexity, the objective remains invariant under any phase rotation of $\mathbf{w}$.
Given this, without loss of optimality, the above problem can be further transformed into the following equivalent form:
\begin{subequations}\label{DL_SNR_Problem_w1}
\begin{align}
\min_{{\mathbf{w}}}&~\delta\| {\mathbf{G}}\mathbf{w}\|_{\rm 2}-\bar{\mathbf{h}}^{\mathsf{H}}{\mathbf{G}}\mathbf{w}\\
{\rm{s.t.}}&~ \Im \left\{ \bar{\mathbf{h}}^{\mathsf{H}}{\mathbf{G}}\mathbf{w} \right\} = 0,\label{dl_c4}\\
&~ \eqref{dl_c3}.\nonumber
\end{align}
\end{subequations}
Finally, by introducing an auxiliary variable $t$, we obtain the epigraph form:
\begingroup              
\setlength{\abovedisplayskip}{-1pt}
\setlength{\belowdisplayskip}{-1pt}
\begin{subequations}\label{DL_SNR_Problem_w2}
\begin{align}
\min_{{\mathbf{w}},t}&~t\\
{\rm{s.t.}}&~\delta\| {\mathbf{G}}\mathbf{w}\|_{\rm 2}-\bar{\mathbf{h}}^{\mathsf{H}}{\mathbf{G}}\mathbf{w} \leq t,\\
&~\eqref{dl_c3},\eqref{dl_c4},\nonumber
\end{align}
\end{subequations}
\endgroup
which is a standard SOCP problem and can be efficiently solved using CVX solvers.

In the lossless case ($\kappa=0$), problem \eqref{DL_SNR_Problem_w} can be simplified through straightforward derivation as follows:
\begin{subequations}\label{DL_SNR_Problem_w_noloss}
\begin{align}
\min_{{\mathbf{w}}}&~\delta\|\mathbf{w}\|_{\rm 2}-\lvert \bar{\mathbf{h}}^{\mathsf{H}}{\mathbf{G}}\mathbf{w}\rvert\\
{\rm{s.t.}}&~ \eqref{dl_c3},\nonumber
\end{align}
\end{subequations}
whose optimal solution of $\mathbf{w}$ can be readily obtained via MRT, i.e., $\mathbf{w}_{\rm MRT}=\frac{\mathbf{G}^{\mathsf{H}}\bar{\mathbf{h}}}{\|\mathbf{G}^{\mathsf{H}}\bar{\mathbf{h}}\|_{\rm2}}\sqrt{P_t}$.

\subsection{Pinching Beamformer}
Given $\mathbf{w}$, problem $\mathcal{P}_{\rm new}$ reduces to:
\begin{subequations}\label{DL_SNR_Problem_P}
\begin{align}
\max_{{\mathbf{P}}}&~\lvert \bar{\mathbf{h}}^{\mathsf{H}}({\mathbf{P}}){\mathbf{G}}({\mathbf{P}})\mathbf{w}\rvert - \delta\| {\mathbf{G}}({\mathbf{P}})\mathbf{w}\|_{\rm 2}\label{DL_SNR_Problem_P_obj}\\
{\rm{s.t.}}&~ \eqref{dl_c1},\eqref{dl_c2}.\nonumber
\end{align}
\end{subequations}
Since the objective function is multimodal, obtaining a globally optimal solution is challenging. A GS1D search method can be employed to sequentially update each PA's position and obtain a suboptimal solution. To reduce the computational overhead of 1D search as much as possible, the subproblem is reformulated as\footnote{To keep the formulation applicable to any channel estimation methods, we assume the estimated CSI is given and do not impose an explicit channel model (e.g., pure LoS, NLoS or mixed LoS/NLoS).}:
\begin{subequations}\label{DL_SNR_Problem_P1}
\begin{align}
\max_{p_{m,n}}&\!\left\lvert \bar{h}_{m,n}\eta_{m,n}^{\frac{1}{2}}{\rm{e}}^{-{\rm{j}}k_{\rm{g}}\left\lvert p_{m,n}-o_m\right\rvert} w_m\! + \!C_1\right\rvert\! - \!\delta\sqrt{\eta_{m,n}\lvert w_m \rvert^2 \!+\!C_2}\label{fpmn}\\
{\rm{s.t.}}& \eqref{dl_c1},\eqref{dl_c2},\nonumber
\end{align}
\end{subequations}
where $C_1 = \sum_{n' \neq n}\bar{h}_{m,n'}\eta_{m,n'}^{\frac{1}{2}}{\rm{e}}^{-{\rm{j}}k_{\rm{g}}\left\lvert p_{m,n'}-o_m\right\rvert} w_m + \sum_{m' \neq m}\bar{\mathbf{h}}_{\rm f}^{\mathsf{T}}({\mathbf{p}}_{m'}){\mathbf{g}}({\mathbf{p}}_{m'})w_{m'}$ and $C_2 = \lvert w_m \rvert^2\sum_{n'\neq n}\eta_{m,n'}+\sum_{m'\neq m}\|{\mathbf{g}}({\mathbf{p}}_{m'})w_{m'}\|_{\rm 2}^2$ denote constants, with $w_m$ being the $m$th entry of $\mathbf{w}$.
In the lossless case ($\kappa=0$), the second term in \eqref{DL_SNR_Problem_P} becomes constant w.r.t. $\mathbf{P}$ and can be dropped. Accordingly, the subproblem for updating $p_{m,n}$ simplifies to maximizing only the first term in \eqref{fpmn}, with $\eta_{m,n} = \frac{1}{N},~\forall m,n$.

The implementation details of 1D search are provided as follows.
In the continuous case, each waveguide is discretized into $N_s$ sampling points, leading to a sampling interval of $\delta_m = L_m/(N_s - 1)$. As $N_s \to \infty$, this scheme approximates a continuous activation. The set of candidate positions along waveguide $m$ is defined as
\begin{equation}\label{set_position}
\mathcal{X}_m \triangleq \left\{o_m + i\delta_m \;\middle|\; i \in \mathcal{I}_m \right\},\forall m\in[M],
\end{equation}
where $\mathcal{I}_m =\{ 0,1,\dots,N_s - 1\}$ denotes the all candidate indices. 
In the discrete case, the PAs' positions are constrained to a pre-configured set of $Q$ locations, corresponding to $\delta_m = L_m/(Q-1)$, and $\mathcal{X}_m$ follows the same structure as in \eqref{set_position} with appropriate $N_s=Q$.

For both cases, the position $p_{m,n}$ is updated by solving
\begingroup              
\setlength{\abovedisplayskip}{1pt}
\setlength{\belowdisplayskip}{1pt}
\begin{equation}\label{position}
p_{m,n}= \arg\max_{p_{m,n} \in \mathcal{X}_m/\mathcal{X}_m(\hat{\mathcal{I}}_m)} \eqref{fpmn},
\end{equation}
\endgroup
where $/$ denotes the set difference, and $\mathcal{X}_m(\hat{\mathcal{I}}_m)$ refers to the subset of $\mathcal{X}_m$ indexed by $i \in \hat{\mathcal{I}}_m$. The index set $\hat{\mathcal{I}}_m$ is defined as follows:
\begin{equation}
\hat{\mathcal{I}}_m \!=\!
\begin{cases}
  \displaystyle
  \mathcal{I}_m \!\cap\!\bigcup_{n'=1}^{n-1} 
  \left\{ i \!\;\middle|\;\!
  i \!\in \!\left\{ i_{m,n'}^{\rm floor}, 
  \dots, 
  i_{m,n'}^{\rm ceil} \right\} 
  \right\},\! & n \neq 1 \\
  \varnothing,\! & n=1,
\end{cases}
\end{equation}
where $i_{m,n'}^{\rm floor}=\left\lfloor \frac{p_{m,n'}-o_m-\Delta_{\rm min}}{\delta_m} \right\rfloor$, $i_{m,n'}^{\rm ceil}=\left\lceil \frac{p_{m,n'}-o_m+\Delta_{\rm min}}{\delta_m} \right\rceil$, ensuring that the spacing between any two PAs on the same waveguide satisfies constraint \eqref{dl_c1}. By sequentially updating all $\{p_{m,n}, \forall m,n\}$ by using \eqref{position}, a locally optimal solution to \eqref{DL_SNR_Problem_P1} can be obtained. 
Since the objective value of $\mathcal{P}_{\rm new}$ is non-decreasing in the alternating updates of $\mathbf{w}$ and $\mathbf{P}$, convergence to a stationary point is guaranteed. 
In the lossy setting, the per-iteration complexity comprises: (i) updating $\mathbf{w}$ via \eqref{DL_SNR_Problem_w2}, which involves $M{+}1$ variables and two SOC constraints of dimensions $MN$ and $M$ and, by \cite{ben2001lectures}, requires $\mathcal{O}(M^{3}N^{2})$ operations; and (ii) updating $\mathbf{P}$, which entails a 1D search of $\mathcal{O}(N_sMN)$ plus computing the constants $C_1$ and $C_2$ with $\mathcal{O}(M^{2}N)$. Therefore, the overall per-iteration complexity is
$\mathcal{O}\!\left(M^{3}N^{2} + N_sMN + M^{2}N\right)$.
In the lossless setting, the $\mathbf{w}$-update reduces to a closed-form MRT beamformer with negligible cost, so the per-iteration complexity is dominated by the $\mathbf{P}$-update and becomes $\mathcal{O}\!\left(N_sMN + M^{2}N\right)$.


\section{Numerical Results}
This section presents simulation results to validate the effectiveness of the proposed algorithm. Unless stated otherwise, parameters are configured as follows. Each of the $M = 4$ waveguides has a length $L_m = D_x = 50~\mathrm{m}$ and is deployed at a height $a = 5~\mathrm{m}$, with uniform spacing $d = \frac{D_y}{M - 1}$ along the $y$-axis, where $D_y = 6~\mathrm{m}$. Each waveguide simultaneously activates $N = 4$ PAs.
The system operates at $f = 28~\mathrm{GHz}$ with transmit power $P_t = 0~\mathrm{dBm}$, noise power $\sigma^2 = -90~\mathrm{dBm}$, and minimum PA spacing $\Delta_{\rm min} = \lambda/2$. The waveguide wavelength is $\lambda_g = \frac{\lambda}{n_{\rm eff}}$, where $n_{\rm eff} = 1.4$, and the attenuation factor is set to $\kappa = 0.08~\mathrm{dB/m}$ \cite{ding2024flexible}.
For continuous activation, we use $N_s = 10^5$ samples; for discrete activation, $Q = 100$ pre-defined locations. Initial PA positions are randomly generated within the feasible region, and the algorithm's performance is insensitive to this initialization.
The normalized channel error bound is defined as $\bar{\delta}=\frac{\delta}{\|\bar{\mathbf{h}}(\mathbf{P})\|_{\rm 2}}$ and set to $\bar{\delta} = 0.3$\footnote{In simulations, $\delta$ is determined based on the estimated channel at the initial PA positions. Alternatively, both systems may use a common $\delta$ derived from the fixed-antenna system.}.
The fixed-antenna system adopts a hybrid transceiver architecture with $M$ RF chains, each connected to $N$ antenna elements via phase shifters. Analog and digital beamformers are designed according to \cite{Zhang2024hybrid}.
All results are averaged over $100$ independent random realizations.

\begin{figure}[!t]
 \centering
\setlength{\abovecaptionskip}{0pt}
\includegraphics[height=0.2\textwidth]{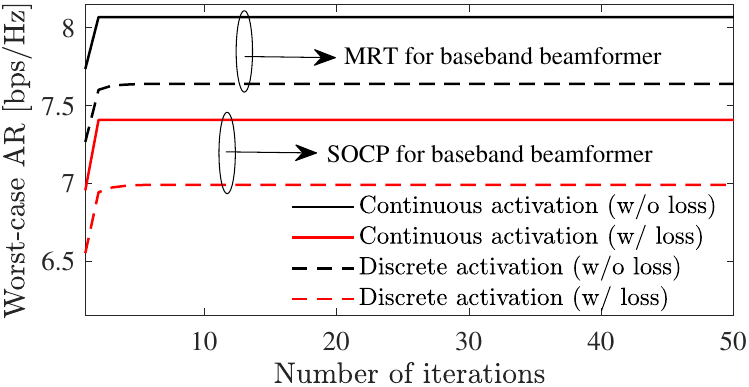}
\caption{The convergence behavior of the proposed algorithms.}
\label{Convergence}
\vspace{-5pt}
\end{figure}
\begin{figure}[!t]
  \centering
  \setlength{\abovecaptionskip}{0pt}
  \subfigure[Continuous activation.]{
    \includegraphics[height=0.18\textwidth]{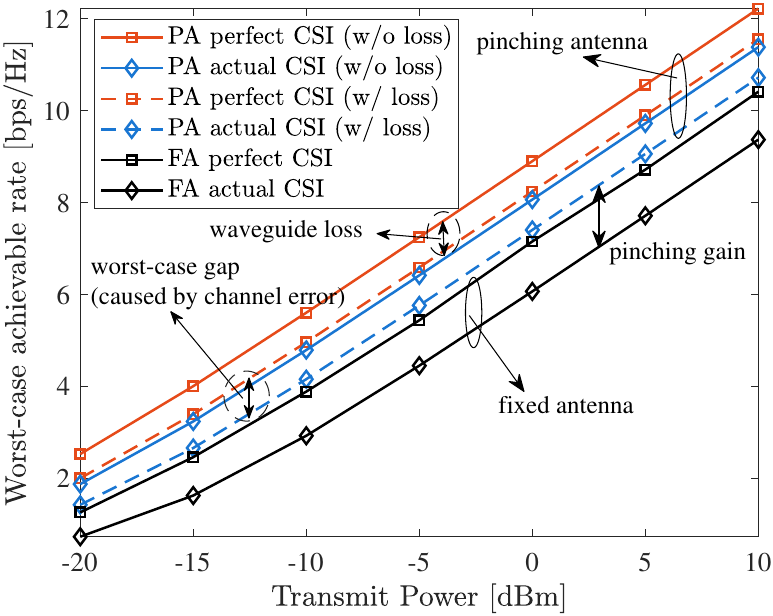}
    \label{Continuous}
  }
  \hspace{-0.1cm}
  \subfigure[Discrete activation.]{
    \includegraphics[height=0.18\textwidth]{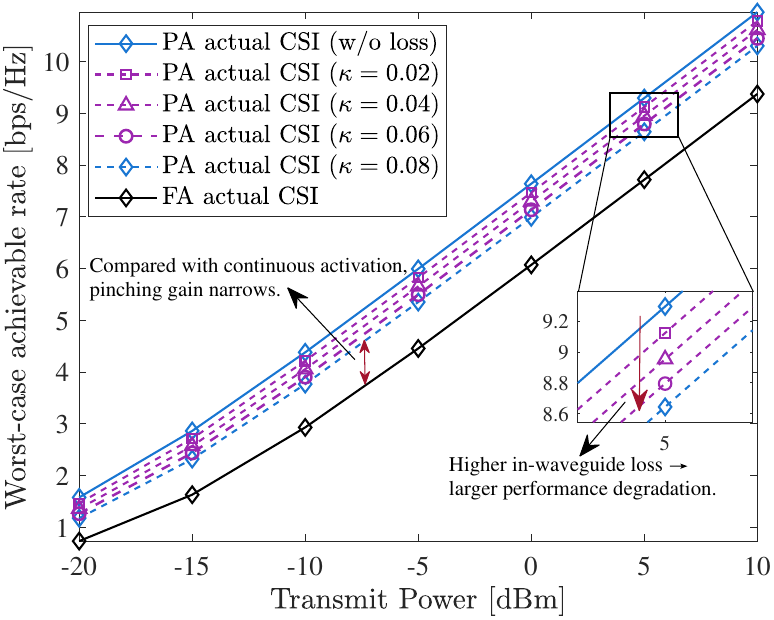}
    \label{Discrete}
  }
  \caption{Worst-case AR versus different transmit power.}
  \label{fig:DL_power}
  \vspace{-5pt}
\end{figure}
\begin{figure}[!t]
  \centering
  \setlength{\abovecaptionskip}{0pt}
  \subfigure[]{
    \includegraphics[height=0.18\textwidth]{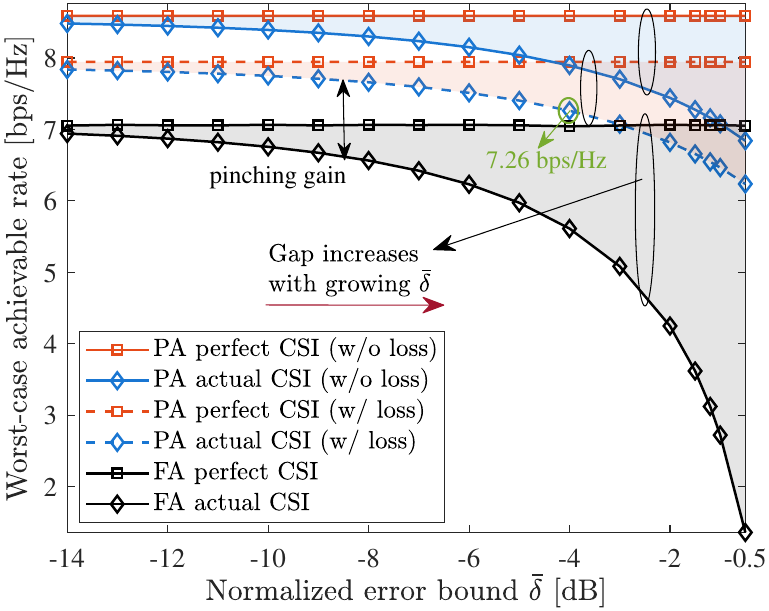}
    \label{error_bounds}
  }
  \hspace{-0.1cm}
  \subfigure[]{
    \includegraphics[height=0.18\textwidth]{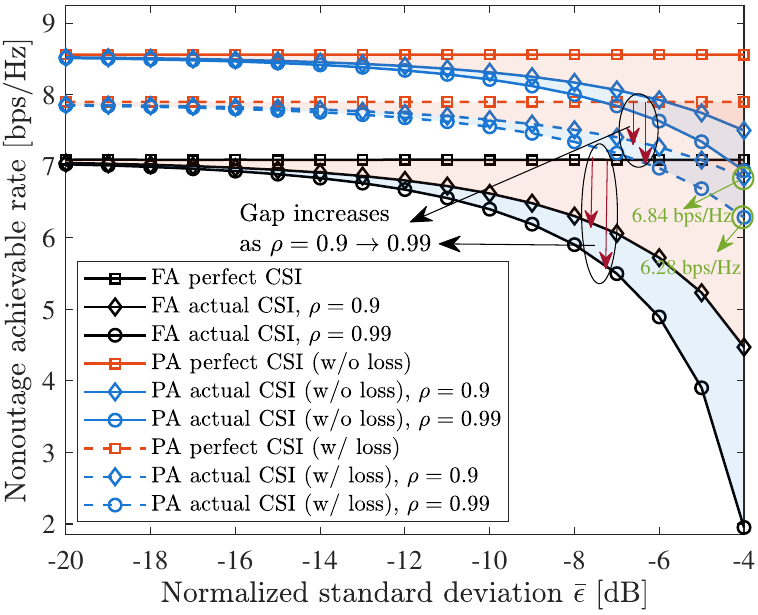}
    \label{variance}
  }
  \caption{Worst-case AR versus different normalized error bounds $\bar{\delta}$ and standard deviations $\bar{\epsilon}$.}
  \label{fig:error}
  \vspace{-5pt}
\end{figure}

Fig.~\ref{Convergence} illustrates the convergence behavior of the proposed algorithm. It converges within five iterations under both lossy and lossless conditions, validating its effectiveness. 
The discrete activation exhibits a relatively slower convergence rate, which is attributed to the lower precision in tuning the signal phase and amplitude.

Fig. \ref{Continuous} and Fig. \ref{Discrete} show the worst-case ARs versus transmit power for continuous and discrete PASS, respectively. As the transmit power increases, all schemes see improved performance. While waveguide loss degrades PASS performance, it still clearly outperforms the fixed-antenna system. For instance, in Fig. \ref{Continuous}, the worst-case AR of PASS surpasses the AR of the fixed-antenna system even in its perfect-CSI case, highlighting PASS's superiority in enhancing signal strength. In contrast, Fig. \ref{Discrete} shows reduced gains with discrete activation due to coarse PA position tuning, which limits phase and amplitude precision. Moreover, It also illustrates that the worst-case AR decreases as the in-waveguide loss increases, due to the attenuation of the received signal power.



Fig. \ref{error_bounds} shows how the worst-case AR varies with the normalized channel error bound $\bar{\delta}$. As expected, the gap between perfect and worst-case CSI increases with growing $\bar{\delta}$ for all systems. Within the considered range of $\bar{\delta}$, PASS consistently achieves a higher worst-case AR than the fixed-antenna system under perfect CSI, demonstrating its effectiveness in practical deployments. Moreover, compared to the fixed-antenna system, the worst-case AR of PASS degrades more slowly. This is because, under a given $\delta$, PASS can adjust PAs' positions to be as close to the user as possible, thereby improving the channel gain. This confirms the superior robustness of PASS in mitigating the impact of channel estimation errors.

Fig. \ref{variance} illustrates the variation of the nonoutage AR under the probabilistic error model with respect to the normalized standard deviation $\bar{\epsilon}$. Similarly, $\bar{\epsilon}$ is defined with $\bar{\epsilon}=\frac{\epsilon}{\|\bar{\mathbf{h}}(\mathbf{P})\|_{\rm 2}}$. It can be observed that Fig. \ref{variance} exhibits a similar decreasing trend as Fig. \ref{error_bounds}. 
However, when $\bar{\epsilon}$ and $\bar{\delta}$ take the same value (e.g. $\bar{\epsilon}=\bar{\delta}=-4$ dB), the nonoutage AR decreases more significantly than the worst-case AR in Fig.~\ref{error_bounds}. This is consistent with \textbf{Lemma~\ref{lemma1}}, which establishes the scaling relationship between $\delta$ and $\epsilon$ via a factor $\sqrt{-\ln(1-\rho)}$. In particular, $\epsilon$ is effectively amplified when the nonoutage probability satisfies $\rho > 1 - \tfrac{1}{e} \approx 0.37$.
In addition, as $\rho$ increases, the nonoutage AR decreases, which is expected, since reliability and efficiency inherently trade off in communication systems.

\section{Conclusion}
This paper proposed a robust beamforming algorithm for PASS under imperfect channel conditions. It was shown that the optimization problems corresponding to both the norm-bounded and probabilistic error models could be equivalently transformed, allowing them to be addressed within a unified AO framework. This equivalence also provided an intuitive understanding of the relationships among the norm bound $\delta$, the standard deviation $\epsilon$, and the nonoutage probability $\rho$. For example, increasing $\delta$ could be interpreted as improving the system's nonoutage probability for a given $\epsilon$. Simulation results validated the effectiveness of the proposed algorithm and further demonstrated that, compared to fixed-antenna systems, PASS enhanced channel gain through pinching beamforming while exhibiting greater robustness against increasing channel uncertainties. Robust multiuser beamforming problems, such as weighted sum-rate maximization and max-min fairness, are considerably more complex, and a thorough treatment is left for future work.

\appendices
\section{Proof of \textbf{Lemma~\ref{lemma1}}}\label{Appendix_A}
Based on the fact that the channel error is not sufficiently large, it can be ensured that $\lvert \bar{\mathbf{h}}^{\mathsf{H}}({\mathbf{P}}){\mathbf{G}}({\mathbf{P}})\mathbf{w}\rvert > \delta\| {\mathbf{G}}({\mathbf{P}})\mathbf{w}\|_{\rm 2}$ holds. We thus begin with $\lvert{\mathbf{h}}^{\mathsf{H}}({\mathbf{P}}){\mathbf{G}}({\mathbf{P}})\mathbf{w}\rvert \geq \lvert \bar{\mathbf{h}}^{\mathsf{H}}({\mathbf{P}}){\mathbf{G}}({\mathbf{P}})\mathbf{w}\rvert - \lvert \mathbf{e}^{\mathsf{H}}({\mathbf{P}}){\mathbf{G}}({\mathbf{P}})\mathbf{w}\rvert$, from which it can be further deduced that 
\begingroup              
\setlength{\abovedisplayskip}{2pt}
\setlength{\belowdisplayskip}{2pt}
\begin{align}
&\text{Pr}\left\{\lvert{\mathbf{h}}^{\mathsf{H}}({\mathbf{P}}){\mathbf{G}}({\mathbf{P}})\mathbf{w}\rvert \geq \Gamma\right\} \\ \nonumber
&\geq \text{Pr}\left\{\lvert \bar{\mathbf{h}}^{\mathsf{H}}({\mathbf{P}}){\mathbf{G}}({\mathbf{P}})\mathbf{w}\rvert - \lvert \mathbf{e}^{\mathsf{H}}({\mathbf{P}}){\mathbf{G}}({\mathbf{P}})\mathbf{w}\rvert \geq \Gamma\right\}.
\end{align}
\endgroup
Accordingly, the original problem $\mathcal{P}_2$ can be rewritten in the following form:
\begin{subequations}\label{DL_SNR_Problem_chance2}
\begin{align}
\max_{{\mathbf{w}},{\mathbf{P}},\Gamma}&~ \Gamma\\
{\rm s.t.}~ & \text{Pr}\left\{\lvert \mathbf{e}^{\mathsf{H}}({\mathbf{P}}){\mathbf{G}}({\mathbf{P}})\mathbf{w}\rvert \leq \lvert \bar{\mathbf{h}}^{\mathsf{H}}({\mathbf{P}}){\mathbf{G}}({\mathbf{P}})\mathbf{w}\rvert \!-\! \Gamma\right\} \!\geq \!\rho\label{nonoutage}\\ 
& \eqref{dl_c3},\eqref{dl_c1},\eqref{dl_c2}.\nonumber
\end{align}
\end{subequations}
To solve this problem, a direct approach is to seek a closed-form expression for the nonoutage probability in constraint \eqref{nonoutage}. Given $\mathbf{e}(\mathbf{P}) \sim \mathcal{CN}(\mathbf{0}, \epsilon^2 \mathbf{I})$, it follows that 
$\mathbf{w}^{\mathsf{H}}{\mathbf{G}}^{\mathsf{H}}({\mathbf{P}})\mathbf{e}(\mathbf{P}) \sim \mathcal{CN}\left(0, \epsilon^2 \|{\mathbf{G}}({\mathbf{P}})\mathbf{w}\|_{\rm 2}^2\right)$. According to the definition of the Rayleigh distribution, the magnitude $|\mathbf{w}^{\mathsf{H}}{\mathbf{G}}^{\mathsf{H}}({\mathbf{P}})\mathbf{e}(\mathbf{P})|$ is Rayleigh distributed, whose cumulative distribution function is given by $F(x) = 1 - \text{exp}\left\{ -\frac{x^2}{\epsilon^2 \|{\mathbf{G}}({\mathbf{P}})\mathbf{w}\|_{\rm 2}^2}\right\}, x>0$. Therefore, we have
\begin{align}
&\text{Pr}\left\{\lvert \mathbf{e}^{\mathsf{H}}({\mathbf{P}}){\mathbf{G}}({\mathbf{P}})\mathbf{w}\rvert \leq \lvert \bar{\mathbf{h}}^{\mathsf{H}}({\mathbf{P}}){\mathbf{G}}({\mathbf{P}})\mathbf{w}\rvert - \Gamma\right\} \\ \nonumber
&= 1- \text{exp}\left\{ -\frac{(\lvert \bar{\mathbf{h}}^{\mathsf{H}}({\mathbf{P}}){\mathbf{G}}({\mathbf{P}})\mathbf{w}\rvert - \Gamma)^2}{\epsilon^2 \|{\mathbf{G}}({\mathbf{P}})\mathbf{w}\|_{\rm 2}^2}\right\}.
\end{align}
Substituting this into \eqref{nonoutage} yields
\begingroup              
\setlength{\abovedisplayskip}{2pt}
\setlength{\belowdisplayskip}{2pt}
\begin{align}
\|{\mathbf{G}}({\mathbf{P}})\mathbf{w}\|_{\rm 2} \leq \frac{\lvert \bar{\mathbf{h}}^{\mathsf{H}}({\mathbf{P}}){\mathbf{G}}({\mathbf{P}})\mathbf{w}\rvert - \Gamma}{\epsilon\sqrt{-\ln(1-\rho)}},
\end{align}
\endgroup
and problem \eqref{DL_SNR_Problem_chance2} can be further reformulated as follows:
\begingroup              
\setlength{\abovedisplayskip}{2pt}
\setlength{\belowdisplayskip}{2pt}
\begin{subequations}\label{DL_SNR_Problem_chance3}
\begin{align}
\max_{{\mathbf{w}},{\mathbf{P}},\Gamma}&~ \Gamma\\
{\rm s.t.}~ & \lvert \bar{\mathbf{h}}^{\mathsf{H}}({\mathbf{P}}){\mathbf{G}}({\mathbf{P}})\mathbf{w}\rvert\!-\!\epsilon\sqrt{\!-\!\ln(1\!-\!\rho)}\|{\mathbf{G}}({\mathbf{P}})\mathbf{w}\|_{\rm 2} \!\geq \!\Gamma\label{nonoutage1}\\ 
& \eqref{dl_c3},\eqref{dl_c1},\eqref{dl_c2},\nonumber
\end{align}
\end{subequations}
\endgroup
which is equivalent to problem $\mathcal{P}_{\rm new}$ under the condition $\delta = \epsilon\sqrt{-\ln(1-\rho)}$, and the proof is thus complete.

\bibliographystyle{IEEEtran}
\bibliography{mybib}
\end{document}